\documentclass[conference]{IEEEtran}
\usepackage{amsfonts}
\usepackage{amssymb}
\usepackage{cite}
\usepackage[cmex10]{amsmath}
\usepackage{float}
\usepackage{color}
\usepackage{stfloats,fancyhdr}
\usepackage{amsmath}
\usepackage{amsthm}
\usepackage{algorithm}
\usepackage{algorithmic}
\usepackage{multirow}
\usepackage{changepage}
\usepackage[normalem]{ulem}
\usepackage{url}

\newtheorem{theorem}{Theorem}

\newtheorem{prob}{Problem}

\IEEEoverridecommandlockouts

\ifCLASSINFOpdf
  \usepackage[pdftex]{graphicx}
  \DeclareGraphicsExtensions{.pdf,.jpeg,.png}
\else
  \usepackage[dvips]{graphicx}
  \DeclareGraphicsExtensions{.pdf}
\fi

\usepackage{subfigure}

\usepackage{fancybox,dashbox}

\begin{document}
%
\title{{Multiuser Scheduling for Minimizing Age of Information in Uplink MIMO Systems}
\thanks{The work of H. Chen is supported by the CUHK direct grant under the project code 4055126. The work of Qian Wang is done when she is a visiting student at The Chinese University of Hong Kong.}
}
\author{\IEEEauthorblockN{He Chen\textsuperscript{1}, Qian Wang\textsuperscript{1,\,2}, Zheng Dong\textsuperscript{3}, and Ning Zhang\textsuperscript{4}
}%
\IEEEauthorblockA{\textsuperscript{1}Department of Information Engineering, The Chinese University of Hong Kong, Hong Kong SAR, China}
\IEEEauthorblockA{\textsuperscript{2}School of Electrical and Information Engineering, The University of Sydney, Sydney, Australia}
\IEEEauthorblockA{\textsuperscript{3}School of Information Science and Engineering, Shandong University, China}
\IEEEauthorblockA{\textsuperscript{4}Department of Computing Sciences, Texas A\&M University-Corpus Christi, TX, USA.\\
Emails: \textsuperscript{1}he.chen@ie.cuhk.edu.hk, \textsuperscript{2}qian.wang2@sydney.edu.au, \textsuperscript{3}zhengdong@sdu.edu.cn, \textsuperscript{4}ning.zhang@tamucc.edu}
}

\maketitle

\begin{abstract}
	This paper studies the user scheduling problem in a multiuser multiple-input multi-output (MIMO) status update system, in which multiple single-antenna devices aim to send their latest statuses to a multiple-antenna information-fusion access point (AP) via a shared wireless channel. The information freshness in the considered system  is quantified by a recently proposed metric, termed age of information (AoI). Thanks to the extra spatial degrees-of-freedom brought about by the multiple antennas at the AP, multiple devices can be granted to transmit simultaneously in each time slot. We aim to seek the optimal scheduling policy that can minimize the network-wide AoI by optimally deciding which device or group of devices to be scheduled for transmission in each slot given the instantaneous AoI values of all devices at the beginning of the slot. To that end, we formulate the multiuser scheduling problem as a Markov decision process (MDP). We attain the optimal policy by resolving the formulated MDP problem and develop a low-complexity sub-optimal policy. Simulation results show that the proposed optimal and sub-optimal policies significantly outperform the state-of-the-art benchmark schemes.
\end{abstract}



\IEEEpeerreviewmaketitle

\section{Introduction}
The concept of age of information (AoI) has recently attracted tremendous attention thanks to its capability of quantifying the information freshness in various applications involving status updates  \cite{Kaul2011mini,kaul2012real,kosta2017age,sun2019age,gu2019timely,gu2019minimizing,wang2019minimizing2,noma,li2020age,chen2020age}. The AoI is defined as the time elapsed since the generation time of the last successfully received status update packet at the destination. Different from the conventional packet-based performance metrics such as delay and throughput, AoI is a new metric that can constantly capture the information staleness from the perspective of the destination.

Early work on AoI resorted to the queueing theory for analyzing the average AoI performance of point-to-point systems with different status update generation models (e.g., generate-at-will and stochastic arrival) and queueing disciplines (e.g., first-come first-serve and last-come first-serve), see \cite{kosta2017age} for a comprehensive survey. Recent efforts on AoI have been shifted to minimize the network-wide AoI of the more practical multi-source networks. Along this research line, those studies that focused on the link scheduling problem are relevant to this work\cite{Pappas2015icc,yates2017status,Joo2017wiopt,Hsu2017isit,Tripathi2017globecom,Sun2018info,Lu2018mobihoc,Ceran2018pimrc,kadota2018scheduling,kadota2018optimizing,kadota2019minimizing,maatouk2020optimality}. In these work, the number of users that can be scheduled to transmit in each time slot is strictly limited by the number of physical orthogonal channels available in the system. For example, if all users share a common channel, at most one of them can be scheduled to transmit in each time~slot to avoid any collision.

In wireless communications, deploying multiple antennas at the transmitter and/or receiver, which is generally referred to as the multiple-input multiple-output (MIMO) technology, has been widely adopted in practical systems like LTE and Wi-Fi to boost the number of users that can be served on each physical frequency channel. This can be realized by leveraging the extra spatial degrees-of-freedom brought about by the multiple antennas. In principle, the maximum number of single-antenna users that can be scheduled to transmit over each physical channel are allowed to be equal to the number of antennas equipped at the receiver side, where various advanced signal processing algorithms can be applied to separate the information sent by multiple users concurrently \cite{spencer2004introduction}. However, an inherent tradeoff exists in multiuser MIMO systems: scheduling more users to transmit in the same time slot will lead to a higher transmission error probability for each scheduled user. The rationale is that each degree-of-freedom associated with the multiple antennas can only be used to either support one more user or boost transmission reliability of those users that have been scheduled to transmit. To our best knowledge, the MIMO technology has not been used to reduce the AoI in multiuser networks.

To fill the gap, in this paper we investigate a multiuser MIMO status update system, in which multiple single-antenna devices want to send their latest statuses to a multiple-antenna access point (AP) via a common wireless uplink channel. Thanks to the multiple antennas equipped at the AP, multiple devices can be scheduled to transmit in each time slot. A fundamental question that arises here is ``\emph{what is the optimal scheduling policy for minimizing the long-term network-wide AoI of the considered system}?". To answer this question, we formulate the multiuser scheduling in the considered system as a Markov decision process (MDP) problem. We obtain the optimal scheduling policy by resolving the formulated MDP. An action elimination procedure is executed to reduce the computational complexity. Furthermore, a sub-optimal policy with much lower complexity is also devised. Simulation results are provided to demonstrate the performance superiority of the developed optimal and sub-optimal policies over the benchmarking schemes that always schedule a fixed number of devices. To our best knowledge, this paper serves as the first attempt to apply the multi-antenna technology for reducing the AoI of multiuser networks.

\section{System Model and Problem Formulation}

\begin{figure}
\centering \scalebox{0.35}{\includegraphics{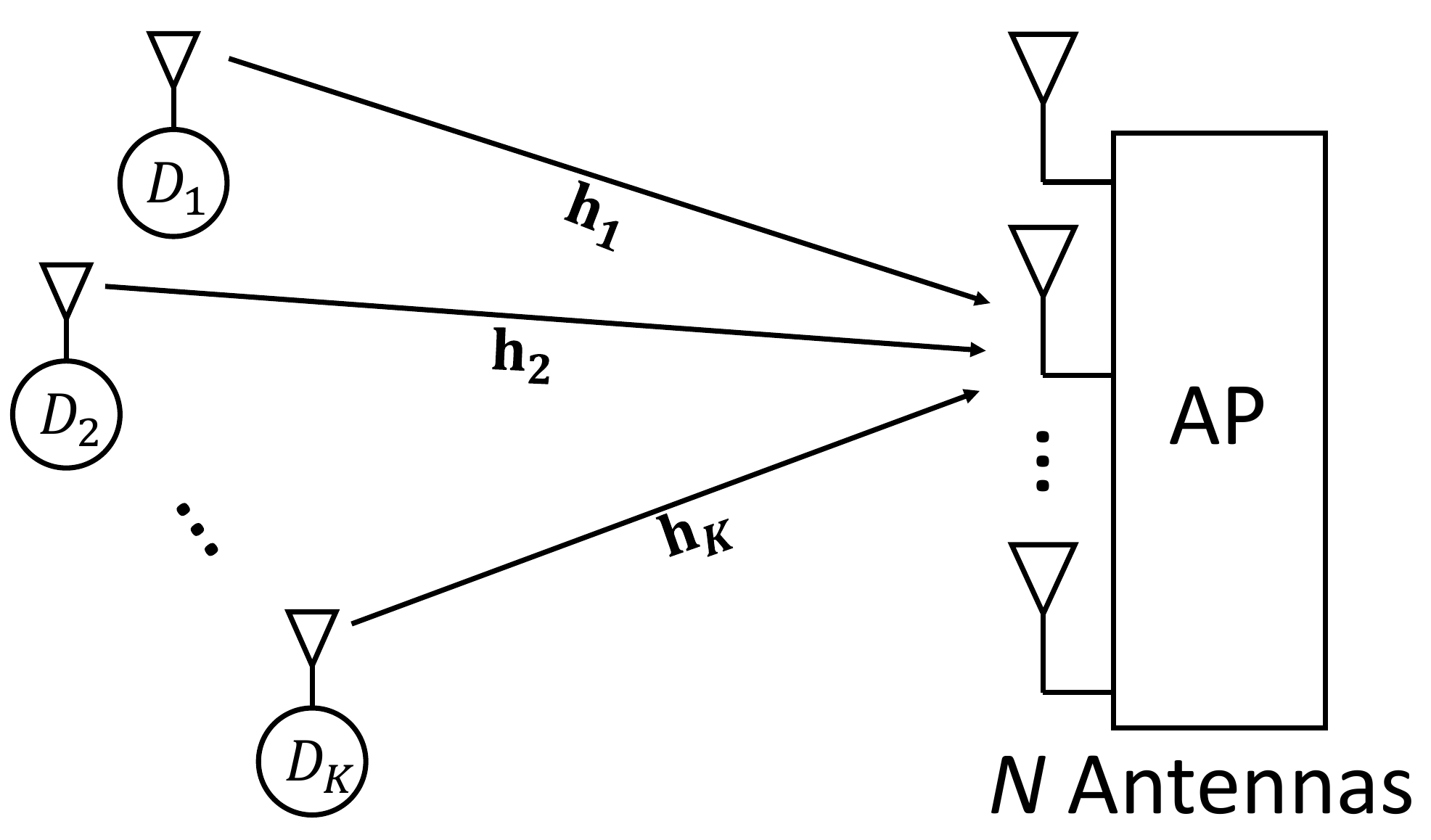}}
\caption{The considered multiuser MIMO status update system with one AP and $N$ end-devices. }\label{fig:systemmodel}
\end{figure}
Consider a multiuser MIMO status update system consisting of an information-fusion AP equipped with $N$ antennas and $K$ end-devices, as depicted in Fig. \ref{fig:systemmodel}. Denote by $D_1,\cdots,D_K$ the $K$ devices, which aim to report their latest statuses to the AP as timely as possible via a shared channel in the uplink. All devices are assumed to be equipped with a single antenna. In this paper, we enforce $N \ge K$ to ensure that the AP can schedule all the devices to access the channel simultaneously if needed. The overloaded case with $N < K$ has been left as a future work. Time is divided into slots of equal durations and the transmission of each status update packet takes exactly one time slot. The timeliness and freshness of the status updates from various devices at the AP is quantified by the recently proposed AoI metric, first coined in~\cite{Kaul2011mini}.

As in \cite{yates2017status,kadota2018scheduling,kadota2018optimizing,kadota2019minimizing}, we consider that the devices are scheduled in a centralized manner and the AP serves as the system coordinator. Specifically, at the beginning of each time slot, the AP will need to decide which device or group of devices will be scheduled to report their latest status(es) in the current time slot. If device $D_i$ is scheduled to transmit in one time slot, it first samples the fresh information and generates a status update packet at the beginning of the time slot, which is known as the ``generate-at-will" model in the literature. Node $D_i$ then sends the generated status update to the AP. Otherwise, $D_i$ will stay idle during the time slot. The principle of the proposed multiuser scheduling algorithm will be elaborated later in this section. Note that in the ``generate-at-will" model, the old status will be replaced by the newly generated one at the beginning of each time slot. As such, in principle each node only needs to maintain a one-packet buffer.

We define\footnote{It is worth pointing out that the multiuser scheduling framework developed in this paper is also applicable for other channel models. We here consider the Rayleigh fading for the sake of the error probability calculation later.} ${\bf{h}}_{i}\sim \mathcal{CN}\left( {{{\bf 0}_N},{\Omega _{i}}{{\bf I}_N}} \right)$ with $i=1,\cdots,K$ to represent the $N\times 1$ channel vector from the device $D_i$ to the AP. We model $\Omega _{i}= d_i^{-\tau}$ with $d_i$ and $\tau$ denoting the distance from the device $D_i$ to the AP and the path loss factor, respectively. Furthermore, we assume quasi-static fading channels for which the fading coefficients remain constant during each transmission slot and change independently from one to another. We assume that all devices are pre-assigned with a unique channel training pilot such that the AP can still estimate the accurate channel states of multiple devices even when they transmit at the same time. As such, the AP can apply various signal processing algorithms to recover the status update packet(s) from the received composite signal. In this paper, we consider that the AP implements the zero-forcing algorithm \cite{Chen2007jsac} to separate the signals sent by multiple devices, thanks to its implementation simplicity.

To elaborate the signal processing at the AP, we now consider the special case that all devices are scheduled to transmit in a certain time slot. Denote by $x_i$ the information sent by device $D_i$ and by $P_i$ the transmit power of device $D_i$. Let ${\bf{n}}\sim {\mathbb C}^{N\times1}$ denote the additive white Gaussian noise with the covariance matrix ${\mathbb E}\left[{\bf nn}^H \right] = \sigma^2 {\bf I}_N$, where ${\mathbb E}[\cdot]$ is the expectation operation and $(\cdot)^H$ is the transpose conjugate. We then can write the received signal at the AP as
\begin{equation}\label{eq:received_signal}
{\bf{y}} = \sum\limits_{i = 1}^K {\sqrt {{P_i}} } {{\bf{h}}_i}{x_i} + {\bf{n}}  = {\bf{Hx}} + {\bf{n}},
\end{equation}
where ${\bf{H}} = \left[ {\sqrt {{P_1}} {{\bf{h}}_1}, \ldots ,\sqrt {{P_K}} {{\bf{h}}_K}} \right]$, ${\bf{x}} = \left[ {{x_1}, \ldots ,{x_K}} \right]^T$ with $[\cdot]^T$ denoting the transpose.

According to the principle of linear zero-forcing receiver \cite{Chen2007jsac}, the information transmitted by multiple nodes will be recovered via
\begin{equation}\label{}
{\bf{\hat x}} = {{\bf{H}}^\dag }{\bf{y}} = {\bf{x}} + {{\bf{H}}^\dag }{\bf{n}},
\end{equation}
where ${{\bf{H}}^\dag } = {\left( {{{\bf{H}}^H}{\bf{H}}} \right)^{ - 1}}{{\bf{H}}^H}$ is the pseudo-inverse of $\bf{H}$. For simplicity, in this paper we concentrate on the symmetric network topology, i.e., $P_1\Omega _{1} = \cdots = P_K\Omega _{K} $. In this case, the error probability for each data stream will only depend on the number of devices scheduled to transmit at the same time. With reference to \cite{Chen2007jsac}, when the AP schedules $k$ out of $K$ devices to transmit simultaneously in the uplink, the achievable error (outage) probability for each scheduled device can be expressed as
\begin{equation}\label{eq:error_prob}
{P_e}\left( k \right) = 1 - \sum\limits_{i = 0}^{N - k} {\frac{{{{\left( {\frac{{{\sigma ^2}}}{{P\Omega }}{\gamma _{\rm {th}}}} \right)}^i}}}{{i!}}\exp \left( { - \frac{{{\sigma ^2}}}{{P\Omega }}{\gamma _{{\rm {th}}}}} \right)},
\end{equation}
where $\gamma _{{\rm {th}}}$ is the required SNR threshold for correct decoding, which is assumed to be identical among all devices. We can see from \eqref{eq:error_prob} that the more the devices to be scheduled in each batch, the higher the error probability for each data stream. The rationale is that more degrees-of-freedom associated with the multiple antennas at the AP is used for cancelling the mutual interference among more data streams.

In the following two subsections, we first formally define the expected AoI of the system, and then mathematically formulate the multiuser scheduling problem.
\subsection{{System Expected AoI}}\label{AoIdefinition}
Denote by $t=1,2,3,\cdots$ the index of time slots and denote by $\delta_i\left(t\right)$, $i\in\left\{1,\cdots,N\right\}$, the instantaneous AoI of the $i$-th device at the beginning time slot $t$. We define the age vector ${\boldsymbol{\delta }}\left( t \right) = \left\{ {{\delta _1}\left( t \right), \ldots ,{\delta _K}\left( t \right)} \right\}$ for notation simplicity. We use $I_i\left(t\right)$ to denote the indicator of whether or not the device $D_i$ is scheduled to transmit in time slot $t$. Particularly, $I_i\left(t\right)=1$ means that $D_i$ will transmit during time slot $t$, and $I_i\left(t\right)=0$ otherwise. Based on the definition of the AoI, the instantaneous AoI of $D_i$ drops to one when $D_i$ successfully delivers a status update to the AP. Otherwise, the instantaneous AoI of $D_i$ increases by one for each time slot. Mathematically, the evolution of the instantaneous AoI for the device $D_i$ can be expressed~as
\begin{equation}
\delta _i \left( {t  + 1} \right) =\left\{
\begin{matrix}
\begin{split}
   &{1}, \quad\text{if}~ I_i \left( t  \right)=1, {~\rm and~} J_i \left( t  \right) = 1, \\
   &{\delta _i \left( {t } \right)+1}, \quad\text{otherwise},  \\
\end{split}
\end{matrix}
\right..
\end{equation}
where $J_i$ is the indicator of whether the transmission of $D_i$ is correct or not when it is scheduled to transmit. $J_i = 1$ if the transmission is correct and $J_i = 0$ otherwise. {\color{black}According to \eqref{eq:error_prob}, we have
\begin{equation}
\label{eq:pr1}
\begin{aligned}
\Pr\left(\delta_i(t+1)=1|\delta_i(t)\right)&=I_i(t)\times\Pr\left(J_i(t)=1\right)\\
&=I_i(t)\left[1-P_e\left(\sum_{i=1}^{K}I_i(t)\right)\right],
\end{aligned}
\end{equation}
\begin{equation}
\label{eq:pr2}
\begin{aligned}
&\Pr\left(\delta_i(t+1)=\delta_i(t)+1|\delta_i(t)\right)=1-I_i(t)\times\Pr(J_i(t) = 1)\\
&~~~~~~~~~~~~~~~~=1-I_i(t)\left[1-P_e\left(\sum_{i=1}^{K}I_i(t)\right)\right].
\end{aligned}
\end{equation}}
{Based on the AoI evolution, the expected AoI of the system can be formally defined as}
{
\begin{equation}\label{AoIexpression}
\bar \delta _s  = \mathop {\lim }\limits_{T \to \infty } {1 \over TK}\sum\limits_{t  = 1}^T \sum\limits_{i  = 1}^K {\delta _i \left( t  \right)}.
\end{equation}}
\subsection{Problem Formulation}
We now describe the concerned age-aware multiuser scheduling problem for considered uplink MIMO system. Specifically, one fundamental question that we want to answer is \emph{for a given age vector ${\boldsymbol{\delta }}\left( t \right) $, which device or group of  devices should be scheduled to transmit in time slot $t$ to minimize the system expected AoI in the long term?} This is actually a non-trivial question since the devices will update more frequently when more of them are scheduled in each batch, which, however, will lead to a higher transmission error probability according to \eqref{eq:error_prob}. We thus aim to find the optimal scheduling policy $\pi$ to minimize the system expected AoI given in \eqref{AoIexpression}. Mathematically, we have the following problem
\begin{prob}\label{problem1}
\begin{equation}\label{}
\mathop {\min }\limits_\pi ~\bar \delta_s(\pi).
\end{equation}
\end{prob}

\section{Optimal and Suboptimal Policies}
To find the optimal multiuser scheduling policy, in this section we recast Problem 1 into an MDP problem, described by a 4-tuple $\{\mathcal{S},\mathcal{A},{P_{\bf a}},r\}$, where
\begin{itemize}
	\item State space $\mathcal{S}\in \mathbb{R}^K$: the state in time slot $t$ is composed by the instantaneous AoI of all clients, ${\bf s}_t \buildrel \Delta \over =  (\delta_{1}(t), ..., \delta_{K}(t))$.
	\item Action space $\mathcal{A}\in\mathbb{R}^{K}$: the action in time slot $t$ is {\color{black}a binary vector ${\bf a}_t=(I_{1}(t), ...,I_{K}(t))$, where $I_i(t)$ is defined in Sec. II.} 
	\item Transition probability ${P}_{\bf a}\left({\bf s},{\bf s}^\prime\right) = \Pr\left({\bf s}_{t+1} = {\bf s}^\prime |{\bf s}_t = {\bf s},{\bf a}_t={\bf a }\right)$ is the transition probability from state ${\bf s}$ to state ${\bf s}^\prime$ when taking action ${\bf a}$ in the time slot $t$. {\color{black}The detailed transition probability is given in \eqref{eq:pr1} and \eqref{eq:pr2}.}
	\item  $r\left({\bf s},{\bf a}\right): \mathcal{S} \times \mathcal{A}  \rightarrow \mathbb{R}$ is the one-stage reward received in time slot $t$, defined as $r\left({\bf s},{\bf a}\right)=\frac{1}{K}\sum_{i=1}^{K}{\bf s}(i)$, which is considered to be independent of the action $\bf a$ in our design. Note that the action $\bf a$ will determine the subsequent state ${\bf s}^\prime$ as well as the corresponding one-stage reward.
\end{itemize}

Given any initial state ${\bf s}_0$, the infinite-horizon average reward of any feasible policy $\pi \in \Pi$ can be expressed as
\begin{equation}
\label{e4}
C(\pi,{\bf s}_0)=\lim_{T \rightarrow \infty}\sup\frac{1}{T} \sum_{t=0}^{T}{\mathbb E}[r_{{\bf{s}}_0}^\pi\left({\bf s}_t,{\bf a}_t\right)].
\end{equation}
Based on the above formulations, we can affirm that the Problem \ref{problem1} can be transformed to the following MDP problem
\begin{prob}
	\label{p2}
	\begin{equation}
	\label{pro2}
	\min_{\pi} C(\pi,{\bf s}_0).
	\end{equation}
\end{prob}
To proceed, we first investigate the existence of optimal stationary and deterministic policy of Problem \ref{p2}. In this regards, we have the following theorem.
\begin{theorem}
	\label{TE1}
	There exist a constant $J^{*}$, a bounded function $h(\mathbf{s}):\mathcal{S} \rightarrow \mathbb{R}$ and a stationary and deterministic policy $\pi^{*}$, satisfies the average reward optimality equation,
	\begin{equation}
	\label{e10}
	J^{*}+h(\mathbf{s})=\min_{{\bf a}\in \mathcal{A}} \left\{r(\mathbf{s},{\bf a}) \mathbb  + {\mathbb{E}}[h({\mathbf{s}}^\prime)]\right\},
	\end{equation} $\forall \mathbf{s} \in \mathcal{S}$, where $\pi^{*}$ is the optimal policy, $J^{*}$ is the optimal average reward, and ${\mathbf{s}}^\prime$ is the next state after $\mathbf{s}$ by executing the action $\bf a$.
\end{theorem}
\begin{proof}
	See Appendix \ref{A1}.
\end{proof}	

{Theorem~\ref{TE1} shows that the optimal policy of the formulated MDP is stationary (i.e., does not vary in time) and deterministic (i.e., is not randomized in action selection). According to \cite[Chapter 8]{sennott2009stochastic}, relative value iteration (RVI) can be used for calculating the optimal policy with a finite state space to approximate that of the countable but infinite state space.}

\subsection{Action Elimination}
The action space for the formulated MDP can be huge, especially when the number of devices in the considered system is large. This is because for a given number of devices $K$, the AP can have $2^K$ possible different actions to take for each state. In addition, the computational complexity of the RVI algorithm is directly related to the sizes of state space and action space. As the state space is composed of the instantaneous AoI of all clients, it cannot be further reduced. In this subsection, we conduct action elimination to reduce the action space for a lower computation complexity, by carefully analyzing the structure of the MDP.

Recall the definition of the one-stage reward function in our MDP formulation, which is the average value of the instantaneous ages of all devices. Hence, to minimize the one-stage reward of the next time slot, the AP should always schedule the devices with larger instantaneous ages. Specifically, if the AP decides to schedule $k$ nodes, it should ask the devices with the maximum $k$ ages to transmit in the current time slot. By doing this, we can reduce the number of possible actions from $2^K$ to $K+1$, which will significantly lower the computation complexity for finding the optimal policy.

{
\subsection{A Sub-optimal Policy}
In this subsection, we propose a low-complexity suboptimal policy, which is inspired by the max-weight policy developed in \cite{kadota2018optimizing,kadota2018scheduling,kadota2019minimizing}. Specifically, the AP chooses the action that minimizes the expected reward of the next state, termed one-step expected next step reward. Given state  ${\bf s}=(\delta_1,\delta_2,...,\delta_{K})$, we sort it in descending order and obtain a new state vector ${\bf s}''=(\delta_1'',\delta_2'',...,\delta_{K}'')$. According to the action elimination section, the number of possible  actions is reduced to $K+1$, denoting how many nodes should be selected for transmission. If the action is to select $k$ nodes for transmission, then the first $k$ nodes in ${\bf s}''$ will be selected and one-step expected next step reward $\mathbb{E}[r({\bf s'}|{\bf s},{\bf a}(k))]$ is
\begin{equation}
\mathbb{E}[r({\bf s'}|{\bf s},{\bf a}(k))]=\frac{1}{K}\left(\sum_{i=1}^{K}\delta_i''+K-(1-P_e(k))\sum_{i=1}^{k}\delta_i''\right).
\end{equation} Specifically, the term $(1-P_e(k))\sum_{i=1}^{k}\delta_i''-k$ is the expected AoI reduction for the $k$ nodes selected for transmission. For the rest $K-k$ nodes, their AoI will all increase by $1$. As such, we have the above one-step expected next step reward $\mathbb{E}[r({\bf s'}|{\bf s},{\bf a}(k))]$. The optimal action for the state ${\bf s}$ under the suboptimal policy $\pi'$ is
\begin{equation}\label{eq:suboptimal}
	\pi'({\bf s})=\arg_{{\bf a}(k), \,k\in\{0,1,...,K\}} \min{\mathbb{E}[r({\bf s'}|{\bf s},{\bf a}(k)])}.
\end{equation}
Compared with the MDP-based optimal policy, the suboptimal policy given in \eqref{eq:suboptimal} is simple to calculate and easy to implement. In addition,  as shown in the numerical results presented in Section IV, the suboptimal policy can achieve a near-optimal performance.
}

\section{Numerical Results and Discussions}
In this section, numerical simulations are provided to compare the performance of the proposed optimal and sub-optimal policies with that of conventional scheduling policies in the considered MIMO status update system over different setups. We set the path loss factor $\tau=2$ and the successful decoding threshold $\gamma _{\rm {th}}=1$ in all simulations. Due to the curse of dimension issue of the RVI algorithm, we consider the scenario with three devices, i.e., $K=3$. In calculating the MDP-based optimal policy, we apply a state truncation ($\delta_i(t) \leq 50$, $\forall i,\;t$) to approximate the countable state space according to \cite{sennott2009stochastic}.
\begin{figure}[!t]
	\centering \scalebox{0.5}{\includegraphics{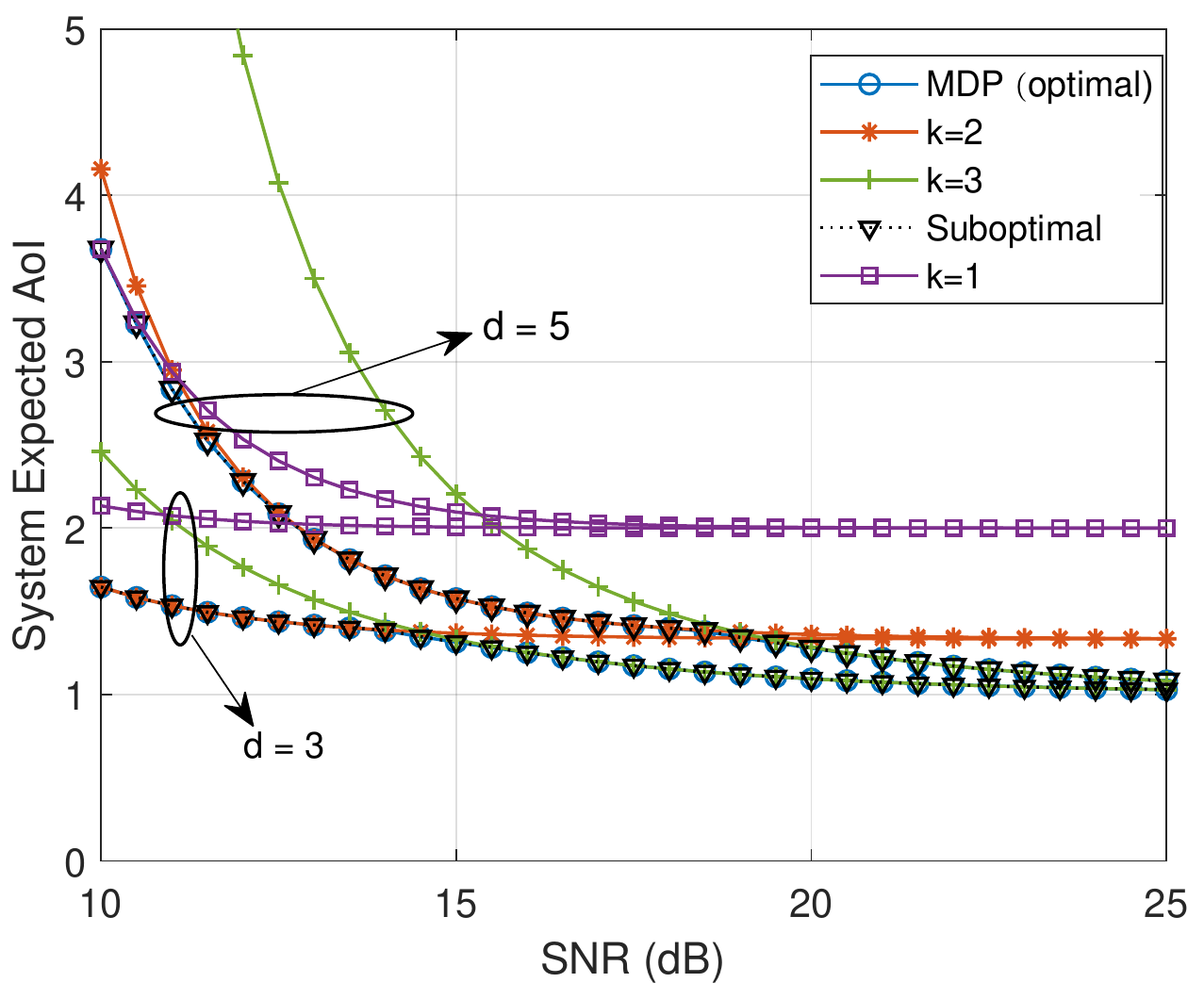}}
	\caption{The performance comparison of different policies versus SNR, when the total number of antennas are $3$, i.e., $N=3$.}
	\label{fig1}
\end{figure}

We plot the system expected AoI curves of various policies versus the transmit SNR (i.e., $P/\sigma^2$) in Fig. \ref{fig1}, in which the number of antennas equipped at the AP is set to $3$. Two groups of curves are plotted corresponding to the cases with $d =3$ and $d=5$, respectively. In each group, the performances of the MDP-based optimal policy, the proposed suboptimal policy given in \eqref{eq:suboptimal}, as well as the stationary policies of scheduling a fixed number of devices ($k=1,\,2,\,3$) are illustrated and compared. We can first see from the figure that for each policy, the system expected AoI is lower when $d$ is smaller. This is understandable since a smaller $d$ indicates a lower error probability for the status update transmissions. However, the performance gap vanishes as the transmit SNR increases. The rationale is that when the SNR is high enough (i.e., the transmission error probability becomes very small), it is the scheduling policy that determines the system performance. We can also observe from Fig. \ref{fig1} that the proposed suboptimal policy can approach the MDP-based optimal policy in all simulated cases. Furthermore, both of them significantly outperform the stationary policies, which always schedule a fixed number of devices.

\begin{figure}[!t]
	\centering \scalebox{0.5}{\includegraphics{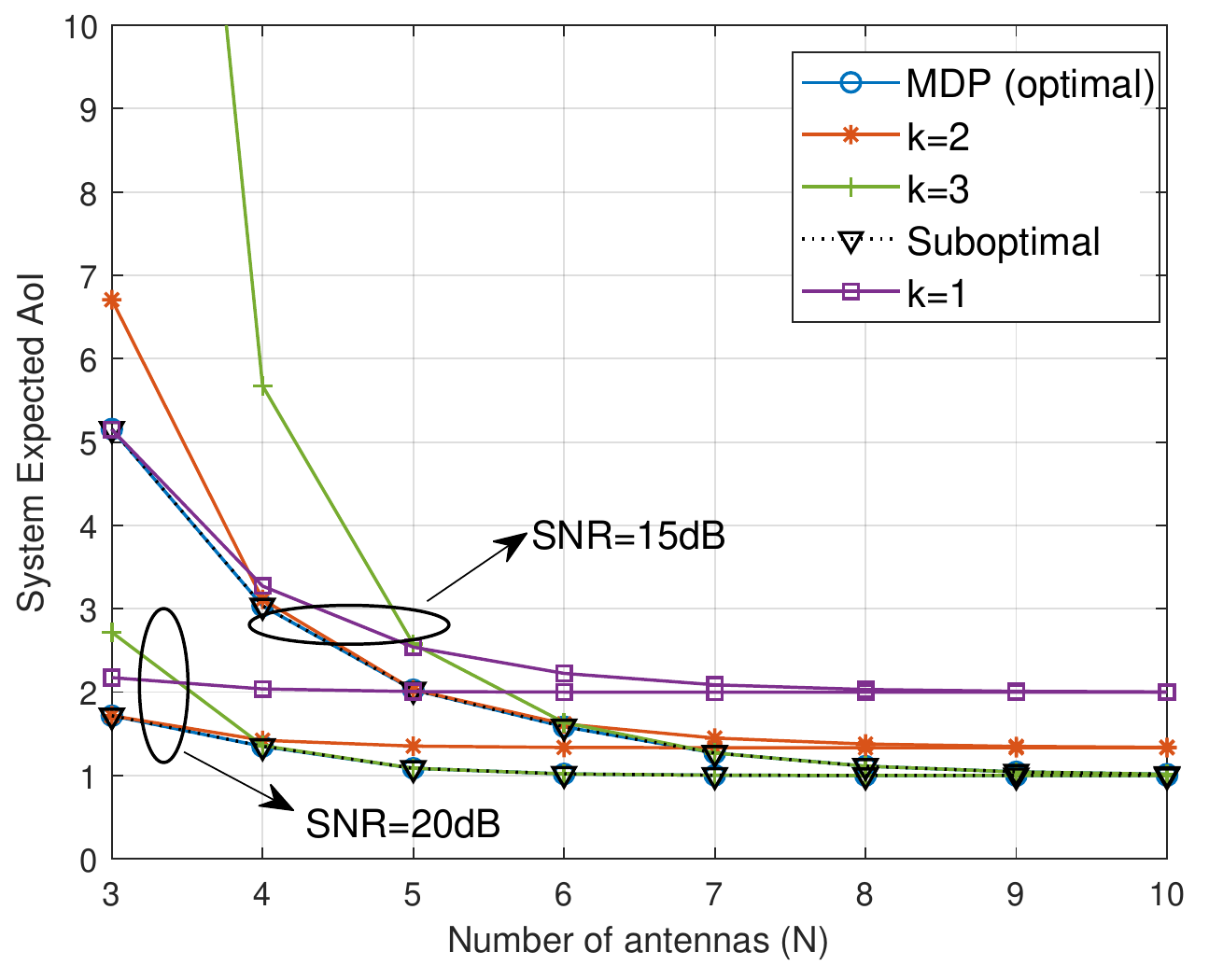}}
	\caption{The performance comparison of different policies versus the number of antennas $N$.}
	\label{fig2}
\end{figure}
We show the performance of all policies considered in Fig. \ref{fig1} versus the number of antennas at the AP (i.e., $N$) for two cases with different transmit SNRs in Fig. \ref{fig2}. Similar phenomenons as in Fig. \ref{fig1} can be observed from Fig. \ref{fig2}. More specifically, the performance of the suboptimal policy in \eqref{eq:suboptimal} almost coincide with that of the optimal policy attained by resolving the formulated MDP. The optimal and sub-optimal adaptive policies developed in this paper are always superior to the stationary policies that schedule a constant number of nodes to transmit in each time slot. The performance of all schemes simulated in Fig. \ref{fig2} tends to be saturated when the value of $N$ is large enough. This is because the transmission error probability decreases as $N$ increases. When $N$ is sufficiently large, scheduling the maximum number of devices (i.e., $k=K$) will lead to the best system performance.



%
\section{Conclusions}
In this paper, we have resolved the problem of minimizing the age of information in multiuser MIMO status update systems. Specifically, we investigated the multiuser scheduling issue and formulated it into a Markov decision process (MDP) problem. We proved that the formulated MDP always admits a stationary and deterministic policy. To reduce the computation complexity of the optimal policy, we executed the action elimination by leveraging a key characteristics of the MDP problem. We further developed a sub-optimal policy that only considers the optimization of the next-step expected reward. Simulation results were provided to demonstrate the superiority of the proposed policies over the benchmarking policies that consistently schedule a fixed number of devices to transmit in each time slot.

\begin{appendices}

	\section{Proof of Theorem \ref{TE1}}
	\label{A1}
	We prove this theorem by verifying the Assumptions 3.1, 3.2 and 3.3 in \cite{guo2006average} hold. As the action space for each state is finite, assumption 3.2 holds, and we only need to verify the following two conditions.
	\begin{itemize}
		\item[\text{1)}] There exist positive constants $\beta < 1$, $G$ and $m$, and a measurable function $\omega({\bf s}) \geq1$ on $S$, ${\bf s}=(\delta_1,\delta_2,...,\delta_{K})$ such that the reward function of MDP problem $r({\bf s},{\bf a})=\frac{1}{K}\sum_{i=1}^{K}\delta_{i}$, $|r({\bf s},{\bf a})|\leq G\omega({\bf s})$ for all state-action pairs $({\bf s},{\bf a})$ and
		\begin{small}
			\begin{equation}
			\sum_{{\bf s}'\in S} \omega({\bf s}')P({\bf s}'|{\bf s},{\bf a})\leq\beta\omega({\bf s})+m,\ {\rm \ for \ all}\ ({\bf s},{\bf a}).
			\end{equation}
		\end{small}
		\item [\text{2)}] There exist two value functions $v_1,v_2 \in B_{\omega}(S)$, and some state ${\bf s}_0\in S$, such that
		\begin{small}
			\begin{equation}
			v_1({\bf s})\leq h_{\alpha}({\bf s})\leq v_2({\bf s}), \ {\rm for \ all} \ {\bf s}\in S,\ {\rm and} \ \alpha \in(0,1),
			\vspace{-.5em}
			\end{equation}
		\end{small}where {\small$h_{\alpha}({\bf s})=V_{\alpha}({\bf s})-V_{\alpha}({\bf s}_0)$} and {\small$B_{\omega}(S):=\{u:\Vert u\Vert_{\omega} <\infty \}$} denotes Banach space, {\small$\Vert u\Vert_{\omega}:=\sup_{{\bf s}\in S}\omega({\bf s})^{-1}|u({\bf s})|$} denotes the weighted supremum norm.
	\end{itemize}
To prove condition 1, we first sort the elements of state ${\bf s}=(\delta_1,\delta_2,...,\delta_{K})$  in descending order and obtain a new age vector ${\bf s}''=(\delta_1'',\delta_2'',...,\delta_{K}'')$, where $\delta_1''$ is the largest element in ${\bf s}$, $\delta_2''$ is the second large element, etc. When the action ${\bf a}$ is to schedule $k$ devices to transmit simultaneously, {the first $k$ devices in ${\bf s}''$ should be scheduled.} In this context, defining $\omega({\bf s})=\frac{1}{K}\sum_{i=1}^{K}\delta_i$, we have $\sum_{{\bf s}'\in S}\omega({\bf s}')P({\bf s}'|{\bf s},{\bf a})=\frac{1}{K}\left(K+\sum_{i=1}^{K}\delta_i''-(1-P_e(k))\sum_{i=1}^{k}\delta_i''\right)$ and $G\geq 1$. Thus, when $\omega({\bf s})=\frac{1}{K}\sum_{i=1}^{K}\delta_i$ and $m >1$, there exist $\max_{k}\{\frac{K+\sum_{i=1}^{K}\delta_i''-(1-P_e(k))\sum_{i=1}^{k}\delta_i''-mK}{\sum_{i=1}^{K}\delta_i}\}\leq\beta<1$ to meet condition 1.

To prove that condition 2 holds in our problem, we show that when $\omega({\bf s})=\frac{1}{K}\sum_{i=1}^{K}\delta_i$, there exists $\frac{\sum_{i=1}^{K}\delta_i+K}{\sum_{i=1}^{K}\delta_i}\leq \kappa<\infty$ that $\sum_{{\bf s}'\in S} \omega({\bf s}')P({\bf s}'|{\bf s},{\bf a})\leq\kappa\omega({\bf s}) $ for all $({\bf s},{\bf a})$, and $\forall {\bf a} \in {\small D^{MD}}$ (Markovian and deterministic (MD) decision rule set), ${\small \sum_{{\bf s}'\in S} \omega({\bf s}')P({\bf s}'|{\bf s},{\bf a})\leq \omega({\bf s})+1\leq (1+1)\omega({\bf s})}$, so that  $\alpha^M\sum_{{\bf s}'\in S} \omega({\bf s}')P_{\pi}^M({\bf s}'|{\bf s},{\bf a})\leq \alpha^M (\omega({\bf s})+M)<\alpha^M(1+M)\omega({\bf s})$, $\pi=({\bf a}_1,...,{\bf a}_M)$. Hence, for each $\alpha$, $0\leq \alpha <1$, there exists a $\eta$, $0\leq \eta<1$ and an integer $M$ such that
	\begin{equation}
	\alpha^M\sum_{{\bf s}'\in S} \omega({\bf s}')P_{\pi}^M({\bf s}'|{\bf s},{\bf a})\leq \eta \omega({\bf s})
	\end{equation} for $\pi=({\bf a}_1,...,{\bf a}_M)$, where ${\bf a}_m\in D^{MD}$, $1\leq m\leq M$. Then, according to \cite[Proposition 6.10.1]{puterman2014markov}, for each $\pi\in \Pi^{MD}$ (MD policy) and ${\bf s}\in S$
	\begin{small}
		\begin{equation}
		|V_{\alpha}({\bf s})|\leq \frac{1}{1-\eta}[1+\alpha\kappa+...+(\alpha\kappa)^{(M-1)}]\omega({\bf s}).
		\vspace{-.5em}
		\end{equation}
	\end{small} We thus can further prove the condition 2. This completes the proof.
\end{appendices}


\ifCLASSOPTIONcaptionsoff
  \newpage
\fi

\bibliographystyle{IEEEtran}
\bibliography{References}

\end{document}